\documentclass[12pt]{article}
\pdfoutput=1
\usepackage{amsmath,amssymb,amsthm,bm,graphicx,float,array,multirow,multicol,rotfloat,caption,subcaption,hyperref,cleveref,enumerate,geometry,mathdots,adjustbox,booktabs,parskip,mathtools,tikz,pdflscape,csquotes,lscape,rotating,empheq,braket,tikz-cd}
\usepackage[para]{threeparttable}
\usepackage[all]{xy}
\usepackage[normalem]{ulem}
\usepackage[numbers,sort&compress]{natbib}
\numberwithin{equation}{section}
\numberwithin{figure}{section}
\allowdisplaybreaks
\makeatletter
\usetikzlibrary{arrows, positioning, decorations.pathmorphing, decorations.markings, decorations.pathreplacing, decorations.markings, matrix, patterns}
\hypersetup{colorlinks=true}
\hypersetup{linkcolor=black}
\hypersetup{citecolor=black}
\hypersetup{urlcolor=black}
\makeatletter

\theoremstyle{plain}
\newtheorem*{thm*}{Theorem}
\newtheorem{thm}{Theorem}[section]

\newtheorem{lem}[thm]{Lemma}

\theoremstyle{definition}

\newtheorem*{defn*}{Definition}

\crefname{lemma}{lemma}{lemmas}
\Crefname{lemma}{Lemma}{Lemmas}
\crefname{thm}{theorem}{theorems}
\Crefname{thm}{Theorem}{Theorems}
\crefname{defn}{definition}{definitions}
\Crefname{defn}{Definition}{Definitions}

\DeclarePairedDelimiterX{\abs}[1]{\lvert}{\rvert}{\ifblank{#1}{{}\cdot{}}{#1}}
\makeatother

\newtheorem*{thm:main}{Theorem \ref{thm:main}}
\newtheorem*{thm:prop}{Proposition \ref{thm:prop}}

\begin{document}

\begin{titlepage}
\vspace*{-3cm} 
\begin{flushright}
{\tt MIT-CTP/5927}\\
\end{flushright}
\begin{center}
\vspace{2.2cm}
{\LARGE\bfseries A no-go  theorem for large $N$ closed universes}\\
\vspace{1cm}
{\large
Elliott Gesteau\\}
\vspace{.6cm}
{ Center for Theoretical Physics -- a Leinweber Institute,\\ Massachusetts Institute of Technology, Cambridge, MA 02139, USA}\par\vspace{.3cm}
{Center of Mathematical Sciences and Applications,\\ Harvard University, Cambridge, MA 02138, USA}
\vspace{.4cm}

\scalebox{.95}{\tt  egesteau@mit.edu}\par
\vspace{1cm}
{\bf{Abstract}}\\
\end{center}
Under conservative assumptions, it is established at a mathematical level of rigor that if correlation functions of single trace operators in a sequence of states of $O(N^0)$ energy of a two-sided holographic conformal field theory admit a large $N$ limit, then they must be described a pure state in the large $N$ Hilbert space of free field theory on two copies of vacuum AdS. This result clarifies recent discussions concerning the possible emergence of a semiclassical baby universe in the large $N$ limit of a low-energy partially entangled thermal state. The proof heavily relies on dominated convergence arguments from mathematical analysis. Some comments are also offered on the relationship between various recently proposed ways of restoring the emergence of the baby universe and the relaxation of different assumptions of the theorem.
\\
\vfill 
\end{titlepage}

\tableofcontents
\newpage
\section{Introduction}

Recent developments in holography increasingly strongly suggest that the Hilbert space of a closed universe is one-dimensional \cite{MarMax20,McNVaf20,AlmMah19a,PenShe19,UsaZha24,UsaWan24}. The consequences are drastic: in particular, a closely related statement is that baby universes cannot emerge from a CFT through the AdS/CFT correspondence, at least in a standard way \cite{AntRat24,EngGes25a}. This strong puzzle has been leading to exciting developments: for example, it has been proposed that in order to make sense of semiclassical physics in a closed universe, we need to explicitly take into account the presence of an observer \cite{HarUsa25,AbdSte25,AkeBue25,EngGes25b,AntSas25}. This set of ideas represents an exciting path forward in the direction of applying the lessons learned from the AdS/CFT correspondence to quantum cosmology.

Extroardinary claims require extroardinary evidence. The recent ``observer-based" proposals involve a significant deviation from quantum mechanics \cite{HarUsa25}, and from the standard rules of the gravitational path integral \cite{AbdSte25}. Therefore, it is very important to make sure that no loophole in the arguments that have been leading us to conclude that such modifications are necessary has been overlooked. With this in mind, this paper aims at providing a mathematically rigorous analysis, and confirmation, of one such family of arguments.

An exciting setup to study closed universes in the context of the AdS/CFT correspondence was proposed and analyzed at length by Antonini, Swingle and Sasieta \cite{AntSas23} (see also \cite{AntSim23,SahVan24}). The authors of \cite{AntSas23} introduced a low-energy partially entangled thermal state of the form 
\begin{align}
\ket{\Psi}=\frac{1}{\sqrt{Z}}\sum_{A,B}e^{-\frac{\beta_L}{2}E_A}e^{-\frac{\beta_R}{2}E_B}\mathbb{O}_{AB}\ket{A}_L\ket{B}_R,
\label{eq:ASSState}
\end{align}
where $\ket{A}, \ket{B}$ are low energy eigenstates of a holographic CFT, $\mathbb{O}$ is a heavy operator, and the inverse temperatures $\beta_L,\beta_R$ are taken above the Hawking--Page transition.

\begin{figure}
\centering
\includegraphics[scale=0.4]{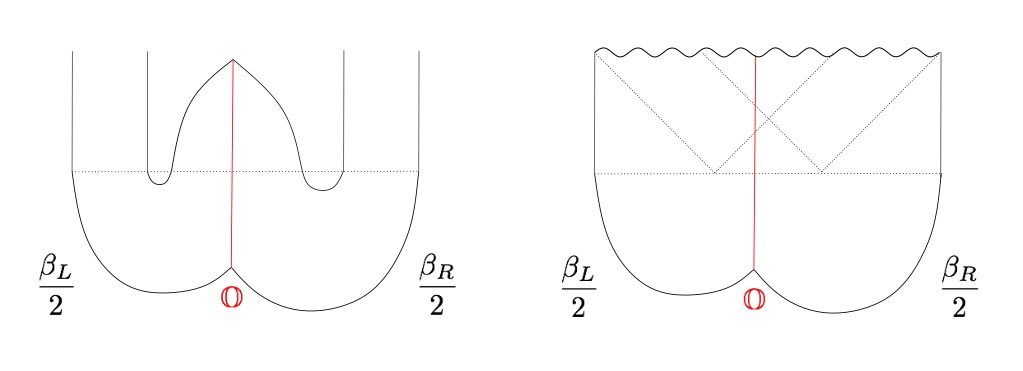}
\caption{Predictions of classical gravity for the dual of the state \eqref{eq:ASSState}. On the left panel, which is the primary case of interest in this paper, the inverse temperatures $\beta_L,\beta_R$ are taken above the Hawking--Page transition (i.e. low temperature), which is the main case of interest of this paper. The solution of Einstein's equation predicts two copies of AdS entangled with a closed universe. On the right panel, the inverse temperatures $\beta_L,\beta_R$ are taken below the Hawking--Page transition (i.e. high temperature). The gravity dual is then thought to be a long wormhole geometry.}
\label{fig:pets}
\end{figure}

If the inverse temperatures $\beta_L,\beta_R$ had been taken below the Hawking--Page transition, the bulk dual to this state would be a long wormhole geometry \cite{GoeLam18}. However, since the present state is of low temperature, a gravitational calculation suggests that the dual consists of two copies of thermal AdS entangled with a closed universe \cite{AntSas23}, see Figure \ref{fig:pets}.

A puzzling aspect of this state was however later pointed out by Antonini and Rath \cite{AntRat24}. Since the state $\ket{\Psi}$ has $O(N^0)$ energy in some regime of parameters, it can be well-approximated by its truncation to a microcanonical window of $O(N^0)$ size in this regime. Such a state is an excitation of the vacuum by a fixed, finite number of light operators, and its dual is therefore believed to be two entangled copies of AdS, with no closed universe to be seen. This led the authors of \cite{AntRat24} to question whether the state $\ket{\Psi}$ really has a unique semiclassical dual. Indeed the state $\ket{\Psi}$ appears to have two descriptions:\newpage
\begin{itemize}
\item \textbf{Description 1:} Two copies of AdS entangled with a closed universe.
\item \textbf{Description 2:} Two copies of AdS in an overall pure state.
\end{itemize}

A resolution of this puzzle was recently proposed by Engelhardt and the author \cite{EngGes25a}. The key idea was that there is a major difference between Descriptions 1 and 2 in that the state of the two copies of AdS is \textbf{mixed} in Description 1, whereas it is \textbf{pure} in Description 2. The strategy of \cite{EngGes25a} was then to introduce a ``swap" operator $\mathcal{S}$ which computes the purity of the causal wedge. This operator can then be reconstructed on the boundary using the standard HKLL map followed by the extrapolate dictionary. Since the boundary expectation value of this reconstruct can be computed and found to be consistent to the bulk expectation value of $\mathcal{S}$ in Description 2, \cite{EngGes25a} concluded that the correct bulk dual of the state $\ket{\Psi}$ is Description 2, which was framed in \cite{EngGes25a} as the statement that the baby universe of Description 1 is not semiclassical, in that there must be a breakdown of the $G\hbar$ expansion in the baby universe. The argument of \cite{EngGes25a} was later reformulated in terms of a holographic code model in \cite{EngGes25b,Hig25}, see also \cite{AntSas25} for related comments.

The sense in which the computation of \cite{EngGes25a} implies that the baby universe is not semiclassical has been extensively discussed in the recent literature. In particular, various ways of restoring the semiclassical description of the baby universe for an internal observer have been proposed, including a resort to a new form of complementarity \cite{EngGes25b}, and a potential method for averaging over the details of the operator $\mathbb{O}$ \cite{AntSas25}. All these approaches have an idea in common: the failure of semiclassicality found in \cite{EngGes25a} should not be interpreted as a failure of semiclassicality to an observer in the baby universe, but rather as a failure of semiclassicality to an external observer who applies the traditional rules of AdS/CFT. 

It is somewhat ambiguous what constitutes the ``traditional rules" of AdS/CFT. In this paper, what will be meant by traditional rules is the notion of large $N$ limit. One of the most rigorous ways to think about AdS/CFT is as follows: given a sequence of finite $N$ CFT states $\ket{\Psi^{(N)}}$, there exists a family of sequences of finite $N$ observables whose expectation values converge when $N\rightarrow\infty$ towards the expectation values of operators of QFT in the emergent bulk background. Some of these sequences are complicated to describe (they correspond to the exponentially complex observables described in qubit models \cite{BroGha19}), but there exists a family of observables that always has a good large $N$ limit: bounded functions of a fixed number of single trace operators.\footnote{The remainder of this paper will make use of a slight abuse of terminology and refer to these as single trace operators.} Given a sequence of finite $N$ single trace operators $S_N$, the large $N$ limit of $\bra{\Psi^{(N)}}S_N\ket{\Psi^{(N)}}$ exists\footnote{More precisely, a one-point function needs to be subtracted in order to ensure the existence of this large $N$ limit.} and corresponds to the expectation value of a free field theory operator on the causal wedge of the emergent background. Causal wedge reconstruction \cite{HamKab06} (or the simplest case of subregion-subalgebra duality \cite{Har16,LeuLiu22}) guarantees that all the elements of the algebra of free field operators on the causal wedge, and their correlations, can be recovered by such large $N$ limits of correlation functions of single trace operators.

The above discussion makes it clear that all correlation functions of fields on the causal wedge of the large $N$ background, and therefore the state of these fields, can be recovered from large $N$ limits of correlation functions of single trace operators. In particular, the large $N$ limit of these correlators should tell us whether the large $N$ causal wedge is in a pure state or in a mixed state.

The goal of this paper is to apply this logic to the state introduced in \cite{AntSas23}, and study the purity of its causal wedge in the large $N$ limit at a mathematical level of rigor. Since it will be necessary to be careful about the distinction between finite and infinite $N$, let us slightly update the expression \eqref{eq:ASSState}. Since the state \eqref{eq:ASSState} is a CFT state, it should really be considered at finite $N$, so let us make this $N$-dependence explicit:
\begin{align}
\ket{\Psi^{(N)}}=\sum_{A,B}c_{AB}^{(N)}\ket{A^{(N)}}_L\ket{B^{(N)}}_R,
\label{eq:finiteN}
\end{align}
where $c_{AB}^{(N)}$ is an $N$-dependent coefficient given by 
\begin{align}
c_{AB}^{(N)}=\frac{1}{\sqrt{Z^{(N)}}}e^{-\frac{\beta_L}{2}E_A^{(N)}}e^{-\frac{\beta_R}{2}E_B^{(N)}}\mathbb{O}_{AB}^{(N)}.
\end{align}
We know that for sequences of finite $N$ single trace operators $S_N$ (or bounded functions thereof), the correlators $\bra{\Psi^{(N)}}S_N\ket{\Psi^{(N)}}$ converge in the large $N$ limit towards those of a free field theory on two copies of AdS. The goal is to determine whether this field theory is in an overall mixed state (which would be the prediction given by Description 1) or in an overall pure state (which would be the prediction given by Description 2). The theorem proven in this paper invalidates the first option, and confirms the second option, as predicted by \cite{EngGes25a}.

Before moving on to stating the theorem itself, it is worth stressing a few of its features and implications.

\begin{enumerate}
\item The theorem gives one new way to think about the precise meaning of the statement of \cite{EngGes25a} that the baby universe is not semiclassical: it does not emerge in the large $N$ limit of the sequence of finite $N$ CFT states \eqref{eq:finiteN}, in that the correlators of single traces in the large $N$ state are those of a pure state.
\item Although the assumptions made in \cite{EngGes25a}, such as the fact that the HKLL map and extrapolate dictionary can be thought of as an approximate isometry, as well as the fact that the bulk effective field theory can be modeled by a finite-dimensional code subspace on which operators of $O(N^0)$ complexity can be mapped to the boundary, are standard and very well-understood in the quantum information literature on AdS/CFT (see for example \cite{FauLi22,Ges23} for precise statements relating holographic codes to the large $N$ limit), one may be tempted to try to find loopholes about these assumptions in the particular case at hand, since the failure of the gravity description involving a closed universe has such dramatic consequences. The theorem presented in this paper makes no such assumption. It does not model AdS/CFT with a qubit code, nor does it assume any isometric property of the HKLL map. It only relies on the large $N$ limit of a \textit{bona fide} conformal field theory, with no simplification of the setup.
\item Since the previous discussions \cite{AntRat24,EngGes25a} of states of the form \eqref{eq:ASSState} were not entirely explicit about the large $N$ limit, their most straightforward interpretation was that they assumed that the coefficients $c_{AB}^{(N)}$ in \eqref{eq:finiteN} are $N$-independent. This may not be true in general and it will be argued in a paper of Liu set to appear on the same day as this paper \cite{Liu25} that in the generic case they may not be. The theorem shown here guarantees that such an $N$-dependence does not alter the conclusion that the causal wedge is in a pure state, under the assumption that the state defined in \eqref{eq:finiteN} admits a large $N$ limit in the usual sense. The only way that such an $N$-dependence can lead to a bulk description with a baby universe is if one relaxes the definition of large $N$ limit. One possible way to do so is to only require it on average, see Section \ref{sec:Discuss} as well as a concrete example and extensive discussion in \cite{Liu25}.
\item The mathematical tools used in the proof teach us some interesting lessons. In particular, it is the fact that dominated convergence arguments can be applied in the present case that prevent the baby universe from emerging below the Hawking--Page temperature.
\item Various proposals for restoring semiclassical physics in the baby universe can be thought of as relaxing different assumptions of the theorem.
\end{enumerate}

The remainder of this paper is organized as follows. In Section \ref{sec:Thm}, after a review on the relevant features of the large $N$ limit and the algebraic approach to AdS/CFT, the precise theorem of this paper is announced. Section \ref{sec:Proof} presents the proof of the theorem.  Various perspectives and possible generalizations are also offered on the result. In particular, it is pointed out that the applicability of various techniques of mathematical analysis is crucial, while these techniques would usually fail in other cases of interest in AdS/CFT. Finally, in Section \ref{sec:Discuss}, a connection is also made between several recent attempts to restore the semiclassicality of the baby universe and the lifting of different assumptions of the theorem.

\section{Theorem}
\label{sec:Thm}
This opening section states the main theorem derived in this paper. In order to better motivate the result and explain its relevance, some background on the algebraic approach to the large $N$ limit of AdS/CFT is also provided.

\subsection{Background: the large $N$ limit and the algebraic approach to AdS/CFT}

The notion of large $N$ limit lies at the core of the original formulation of the AdS/CFT correspondence \cite{Mal97}. It is only in the strict $N=\infty$ limit that an exact classical spacetime background emerges in the bulk. In this limit, all observables in the causal wedge can be described by a generalized free field theory on the boundary, which is an equivalent description of a free quantum field theory in the causal wedge \cite{LeuLiu21b,LeuLiu22,GesLiu24}. 

In a rigorous treatment of the large $N$ limit, the finite $N$ and infinite $N$ Hilbert spaces are treated on a different footing. At finite $N$, the Hilbert space $\mathcal{H}_N$ is the Hilbert space of the CFT at a large but finite value of the central charge. By contrast, at infinite $N$, there is no actual CFT to define the Hilbert space. Instead, the Hilbert space is defined through the correlation functions of operators that admit a large $N$ limit, chief among which are the single trace operators. Consider $\ket{\Psi^{(N)}}$ is a finite $N$ state, and a sequence of finite $N$ single trace operators 
\begin{align}
S_N=\mathrm{Tr}\,\phi_N-\langle \mathrm{Tr}\,\phi_N\rangle.
\label{eq:singletrace}
\end{align}
The sequence of states $\ket{\Psi^{(N)}}$ has a good large $N$ limit if the expressions $\bra{\Psi^{(N)}}S_N\ket{\Psi^{(N)}}$ are convergent as $N\rightarrow\infty$ for all linear combinations of products of normalized single trace operators of the form \eqref{eq:singletrace}.

Note that a one point function has been subtracted in Equation \eqref{eq:singletrace} to ensure the state admits a large $N$ limit \cite{Wit21}. This one point function encodes information about the large $N$ background.

For those sequences of states $\ket{\Psi^{(N)}}$ that admit a large $N$ limit in the above sense, the algebra of observables in the causal wedge and an associated Hilbert space are defined directly at large $N$ through the large $N$ correlation functions of single trace operators. If one state in this Hilbert space satisfies large $N$ factorization the construction can be made more explicit. In particular, in the bosonic case it is generated by Weyl operators of the form \begin{align}W(f)=e^{i\phi(f)},\end{align} where $f$ is a smearing function on the boundary, and the field operators $\phi$ satisfy the relation 
\begin{align}
[\phi(f),\phi(g)]=i\int f(x)E(x,y)g(y),
\end{align}
where $E(x,y)$ is the imaginary part of the large $N$ Wightman function. This algebra of observables is then represented on an infinite $N$ Hilbert space $\mathcal{H}$ whose inner product is engineered so that in the factorizing state of interest $\ket{\Omega}$,
\begin{align}
\langle\Omega\vert\phi(f)\phi(g)\vert\Omega\rangle=\int f(x)G(x,y)g(x)dx,
\end{align}
where $G(x,y)$ is the large $N$ two point function. An important point to be noted here is that these algebras and inner products are sector-dependent, in particular different large $N$ two point functions can lead to different infinite $N$ Hilbert spaces and algebras. The reader is referred to \cite{FauLi22,FurLas23,GesSan24}
for detailed accounts of the mathematical construction described here. The main conceptual point that will be important in what follows is that all the information about the inner product of the infinite $N$ Hilbert space and expectation values of operators on this Hilbert space is contained in large $N$ correlators.

How do we then relate the finite $N$ picture to the infinite $N$ picture? We say that a sequence of states $\ket{\Psi^{(N)}}$ limits to a state $\ket{\Psi}$ in a given large $N$ Hilbert space, if for all sequences $(S_N)$ of finite $N$ single trace operators with the appropriate one-point function subtracted, represented at large $N$ by a free field theory operator $S$ of the corresponding large $N$ algebra, 
\begin{align}
\langle\Psi^{(N)}\vert S_N\vert\Psi^{(N)}\rangle\underset{N\rightarrow\infty}{\longrightarrow} \langle{\Psi}\vert S\vert \Psi\rangle.
\end{align}

The operator $S$ is defined by replacing finite $N$ single trace operators by generalized free field operators in the most straightforward way, for example, at large $N$, $\mathrm{Tr}\,\phi_N-\langle \mathrm{Tr}\,\phi_N\rangle$ gets replaced by a field operator $\phi$.

One large $N$ sector will be of particular interest in this work: the one of two copies of vacuum AdS. Large $N$ convergence towards this sector is conceptually simpler in this case, since the $N$-dependent part of the one point function of most single trace operators vanishes in the vacuum state.\footnote{I am grateful to Hong Liu for clarifications on this point.} Therefore in the case of this Hilbert space, we can think of the $(S_N)$ sequences literally as sequences of finite $N$ single trace operators or bounded functions thereof, with no one point function subtracted.

\subsection{Assumptions, useful facts, and theorem statement}

Let us now state the precise theorem proven in this paper. Consider a sequence of finite $N$ states 
\begin{align}
\ket{\Psi^{(N)}}:=\sum_{A,B}c_{AB}^{(N)}\ket{A^{(N)}}\ket{B^{(N)}},
\end{align}
where the $\ket{A^{(N)}}$, $\ket{B^{(N)}}$ are the energy eigenstates at finite $N$ ordered by increasing order of energy (for simplicity let us assume no degeneracy). For $\Delta>0$ fixed and independent of $N$, we also introduce the truncated states
\begin{align}
\ket{\Psi_\Delta^{(N)}}:=\sum_{E<\Delta}c_{AB}^{(N)}\ket{A^{(N)}}\ket{B^{(N)}}.
\end{align}
The setup is:
\begin{itemize}
\item \textbf{Assumption 1: well-defined large $N$ limit in the causal wedge.} The sequences $\bra{\Psi^{(N)}}S_N\ket{\Psi^{(N)}}$, for $(S_N)$ a sequence of single trace operators or bounded functions thereof, converge when $N\rightarrow\infty$. We will also use the technical assumption that the $S_N$ can be chosen to be \textbf{uniformly bounded} in norm, i.e. bounded in norm by a constant independent of $N$.

\textit{This assumption supposes that the sequence of states $(\ket{\Psi^{(N)}})$ admits a well-defined large $N$ limit in the vacuum sector. The existence of a large $N$ limit is necessary for the emergence of semiclassical physics in the context of traditional discussions of the AdS/CFT correspondence \cite{Wit21b,LeuLiu21b,FauLi22}. Moreover in both descriptions 1 and 2 of the Antonini--Rath state of \cite{AntRat24}, the causal wedge observables, i.e. single trace operators, are present.\footnote{At least if it is assumed that operators that converge observables on a vacuum AdS background never have any one point function subtracted.} The question is to know whether they are in an overall pure or mixed state.
Uniform boundedness is obtained, for example, in the bosonic case by taking exponentials of the single trace operators, see Equation (3.45) of \cite{FauLi22}.}

\item \textbf{Assumption 2: low energy.} The family of states $(\ket{\Psi^{(N)}})$ has uniformly bounded energy in $N$ by a constant $E_0=O(1)$, independent of $N$.

\textit{This assumption is in particular true for the state considered in Equation \eqref{eq:ASSState}, in the regime of parameters of interest here.}
\end{itemize}

Two true facts about holographic conformal field theories will also be used:
\begin{itemize}
\item \textbf{Fact 1: constant vacuum excitations are pure.} If the $c^{(N)}_{AB}$ are \textbf{constant} functions $c_{AB}^0$ of $N$, for generic choices of $\Delta$ the sequences $\bra{\Psi_\Delta^{(N)}}S_N\ket{\Psi_\Delta^{(N)}}$, for $(S_N)$ a sequence of single trace operators asymptoting to an observable $S$ in the vacuum sector, all converge when $N\rightarrow\infty$ towards $\bra{\Psi_\Delta}S\ket{\Psi_\Delta}$, where $\ket{\Psi_\Delta}$ is a \textbf{pure} state in the vacuum sector (up to normalization). Moreover the state $\ket{\Psi_\Delta}$ has the explicit expression in the large $N$ Hilbert space of two copies of vacuum AdS:
\begin{align}
\ket{\Psi_\Delta}=\sum_{E<\Delta}c_{AB}^0\ket{A}\ket{B},
\end{align}
where now $A$ and $B$ label the energy eigenstates in the \textbf{large $N$} vacuum sector. 

\textit{Note that this fact is only obvious if the $c^{(N)}_{AB}$ are \textbf{constant} functions $c_{AB}^0$, see \cite{Liu25} for more on this.}
\item \textbf{Fact 2: spectrum property.} In a sequence of finite $N$ holographic CFTs, the number of states whose energy is upper bounded by a constant $E_0=O(N^0)$ is bounded by a constant $K$ independent of $N$, and this number stabilizes as $N\rightarrow\infty$ for a generic choice of $E_0$. 

\textit{Note that this is in sharp contrast with the high energy part of the spectrum: there are $O(e^{N^2})$ states whith $O(N^2)$ energy \cite{FesLiu06}.}
\end{itemize}

Let us now announce the main theorem of this paper.

\begin{thm}
Under the two assumptions and two facts above, there exists a \textbf{pure} state $\ket{\Psi}$ in the large $N$ Hilbert space of the vacuum such that for all uniformly bounded sequences of observables $(S_N)$ of single trace observables described by a large $N$ bounded field operator $S$, \begin{align}
\bra{\Psi^{(N)}}S_N\ket{\Psi^{(N)}}\underset{N\rightarrow\infty}{\longrightarrow} \bra{\Psi}S\ket{\Psi}.
\end{align}
\label{thm:main}
\end{thm}
\section{Proof}
\label{sec:Proof}
This section presents the proof of Theorem \ref{thm:main}. In order to make it as easy to read as possible, the steps are broken down into several subsections. The end of this section also offers some comments on why this proof works, and the mathematical difference between the setup presented in this paper and some other cases of interest in AdS/CFT.

\subsection{Step 1: bound the error made by finite $N$ truncations}
The first step of the proof, inspired by \cite{AntRat24}, is to bound the error made by truncating the finite $N$ state $\ket{\Psi^{(N)}}$ to a microcanonical window of maximum energy $\Delta$ independent of $N$, which is done in the following lemma. Note that a closely related result is established in \cite{EngGes25a}. In this lemma, the norm on states is the norm on linear functionals (also known as the 1-norm):
\begin{align}
\|\psi_1-\psi_2\|:=\underset{\|X\|=1}{\mathrm{sup}}\lvert\bra{\Psi_1}X\ket{\Psi_1}-\bra{\Psi_2}X\ket{\Psi_2}\rvert.
\end{align}
\begin{lem}
Let $\varepsilon>0$. There exists $\Delta>0$ such that for all $N$,
\begin{align}
    \|\psi_\Delta^{(N)}-\psi^{(N)}\|<\varepsilon.
\end{align}
\label{lem:truncate}
\end{lem}
\begin{proof}
For $\Delta>0$, by using Markov's inequality and our Assumption 2, we obtain
\begin{align}
\|\ket{\Psi^{(N)}_\Delta}-\ket{\Psi^{(N)}}\|^2=\mathbb{P}(E>\Delta)\leq\frac{E_0}{\Delta}.
\end{align}
Then, for $S_N$ a bounded operator on the finite $N$ Hilbert space, by the triangle inequality followed by the Cauchy--Schwarz inequality,
\begin{align}
\nonumber\lvert\bra{\Psi^{(N)}_\Delta}S_N\ket{\Psi_\Delta^{(N)}}-\bra{\Psi^{(N)}}S_N\ket{\Psi^{(N)}}\rvert&\leq \lvert\bra{\Psi^{(N)}_\Delta}S_N\ket{\Psi_\Delta^{(N)}}-\bra{\Psi^{(N)}}S_N\ket{\Psi^{(N)}_\Delta}\rvert+\\&\lvert\bra{\Psi^{(N)}}S_N\ket{\Psi^{(N)}_\Delta}-\bra{\Psi^{(N)}}S_N\ket{\Psi^{(N)}}\rvert\\
\leq 2\|S_N\|\sqrt{\frac{E_0}{\Delta}}.
\end{align}
Choosing $\Delta$ large enough, we get 
\begin{align}
\lvert\bra{\Psi^{(N)}_\Delta}S_N\ket{\Psi_\Delta^{(N)}}-\bra{\Psi^{(N)}}S_N\ket{\Psi^{(N)}}\rvert\leq\varepsilon\|S_N\|,
\end{align}
which proves the lemma.
\end{proof}
\subsection{Step 2: Cantor's diagonal extraction}

One main difficulty is that one cannot a priori exclude the case where the coefficients $c_{AB}^{(N)}$ do not converge to a fixed value as $N\rightarrow\infty$ (see \cite{Liu25} for more on this.) However, there is a weaker statement that is true, which is that there will always exist a subset of values of $N$, $\{\phi(N)\}_{N\in\mathbb{N}}$ such that the $c_{AB}^{\phi(N)}$ all converge. This is due to a classic technical trick in analysis known as Cantor's diagonal extraction procedure. This trick is introduced here, so that it can be used in the next step.

\begin{lem}[Cantor's diagonal extraction]
Let $(c^{(N)}_{AB})$ be a family of sequences of complex numbers, each bounded in $N$, indexed by $AB$ where $AB$ runs over a countable set. Then, there exists an extractor function $\phi(N)$ such that for all $AB$, the sequence $(c^{\phi(N)}_{AB})$ is convergent as $N\rightarrow\infty$.
\label{lem:cantor}
\end{lem}
\begin{proof}
First, fix $AB$. Enumerate the indices $AB$ with an integer-valued bijection $k(AB)$. Since the $c_{k^{-1}(1)}^{(N)}$ are bounded, there exists an extractor function $\phi_1(N)$ for which the $c_{k^{-1}(1)}^{\phi_1(N)}$ converge. The process can be repeated: there exists a subset $\{\phi_2(N)\}$ of the $\{\phi_1(N)\}$ such that the $c_{k^{-1}(2)}^{\phi_2(N)}$ also converge. Iteratively, for all $n\in\mathbb{N}$ there exists an extractor function $\phi_{n}(N)$ such that the $c_{k^{-1}(1)}^{\phi_{n}(N)},\dots,c_{k^{-1}(n)}^{\phi_{n}(N)}$ converge.  We now define the extractor function \begin{align}
  \phi(N):=\phi_{N}(N).  
\end{align}
By construction, all the $(c_{AB}^{\phi(N)})$ converge as $N\rightarrow\infty$.
\end{proof}

\subsection{Step 3: construct a limit for some of the truncated states}

Now, we can use Cantor's diagonal extraction procedure to construct a subset $\{\phi(N)\}_{N\in\mathbb{N}}$ of values of $N$ for which all the $c_{AB}^{\phi(N)}$ converge. In general, it would be hard to deduce from this simple trick that the correlation functions in the states $\ket{\Psi^{\phi(N)}}$ get close to the ones in a state in which the $c_{AB}^{\phi(N)}$ have been replaced by their limits. Here however, at least for states that are truncated to a microcanonical window of finite size, there are only finitely many $c_{AB}^{\phi(N)}$ to control, so that this statement is actually true, as shown by the following lemma. Note that the fact that this step works here, whereas it would fail in many other cases of interest in AdS/CFT, is crucial. This will be discussed in more detail at the end of this section.

\begin{lem}
Let $\Delta>0.$ There exists an extractor function $\phi(N)$ and a pure state $\ket{\Psi_\Delta}$ in the large $N$ vacuum sector such that for all uniformly bounded sequences of finite $N$ single trace operators $(S_N)$ asymptoting to a bounded operator $S$ on the large $N$ vacuum sector, 
\begin{align}
\bra{\Psi_\Delta^{\phi(N)}}S_{\phi(N)}\ket{\Psi_\Delta^{\phi(N)}}\underset{N\rightarrow\infty}{\longrightarrow}\bra{\Psi_\Delta}S\ket{\Psi_\Delta}.
\end{align}
\label{lem:extract}
\end{lem}

\begin{proof}
The $c_{AB}^{(N)}$ are complex numbers upper-bounded in norm by 1. Therefore we can apply the previous lemma to the $c_{AB}^{(N)}$. We construct an extractor $\phi(N)$ such that all the $c_{AB}^{\phi(N)}$ converge in the large $N$ limit to a value $c_{AB}^{0}$. We now define the finite $N$ states
\begin{align}
\ket{\Psi_{\Delta,0}^{(N)}}:=\sum_{E<\Delta}c_{AB}^{0}\ket{A^{(N)}}\ket{B^{(N)}}.
\end{align}
By Fact 1 recalled in the beginning of this section, correlation functions in this state have a well-defined limit, described by some pure state (up to normalization) $\ket{\Psi_\Delta}$ in the large $N$ vacuum sector.
For $(S_N)$ a uniformly bounded sequence of vacuum sector observables described by an infinite $N$ bounded operator $S$, we now estimate (using the notation $\psi(X)$ for $\bra{\Psi}X\ket{\Psi}$):
\begin{align}
\label{eq:convdelta}
|\psi^{\phi(N)}_\Delta(S_{\phi(N)})-\psi_\Delta(S)\rvert\leq\lvert\psi^{\phi(N)}_\Delta(S_{\phi(N)})-\psi_{\Delta,0}^{\phi(N)}(S_{\phi(N)})\rvert+\lvert\psi_{\Delta,0}^{\phi(N)}(S_{\phi(N)})-\psi_{\Delta}(S)\rvert.
\end{align}
Let us estimate the first term.
\begin{align}
\lvert\psi_\Delta^{{\phi(N)}}(S_{\phi(N)})-\psi_{\Delta,0}^{{\phi(N)}}(S_{\phi(N)})\rvert&=\lvert\sum_{E(AB,A^\prime B^\prime)<\Delta}(c_{AB}^{{\phi(N)}\ast}c_{A^\prime B^\prime}^{{\phi(N)}}-c_{AB}^{0\ast}c_{A^\prime B^\prime}^{0})\bra{AB^{\phi(N)}}S_{\phi(N)}\ket{A^\prime B^{\prime{\phi(N)}}}\rvert\\&\leq\|S_{\phi(N)}\|\sum_{E(AB,A^\prime B^\prime)<\Delta}\lvert(c_{AB}^{\phi(N)\ast}c_{A^\prime B^\prime}^{\phi(N)}-c_{AB}^{0\ast}c_{A^\prime B^\prime}^{0})\rvert.
\end{align}
Since $\|S_{\phi(N)}\|$ is uniformly bounded, each term of this sum goes to zero. Since the sum is finite, and the number of terms does not depend on $N$ by Fact 2, we further conclude that this expression goes to zero as $N\rightarrow\infty$. Therefore the first term of \eqref{eq:convdelta} goes to zero in the large $N$ limit. By definition of $\ket{\Psi_\Delta}$, the second term also goes to zero, which concludes the proof of the lemma.
\end{proof}

\subsection{Step 4: construct a limit for some of the untruncated states}

At this stage of the proof, there is an obvious candidate for a potential pure state limit of the $\ket{\Psi^{\phi(N)}}$ in the large $N$ vacuum sector: the limit of the $\ket{\Psi_\Delta}$. It remains to be shown that this limit exists, which is done by the next lemma using the Cauchy criterion.

\begin{lem}
The family $(\ket{\Psi_\Delta})$ is Cauchy for the norm of the large $N$ vacuum Hilbert space.
\label{lem:cauchy}
\end{lem}
\begin{proof}
Let $\Delta_0>\Delta_1>0$. By Fact 1, for $N$ large enough,
\begin{align}
\|\ket{\Psi_{\Delta_0}} - \ket{\Psi_{\Delta_1}}\|^2 
= \sum_{\Delta_1 < E < \Delta_0} \lvert c_{AB}^0 \rvert^2 
= \|\ket{\Psi_{\Delta_0,0}^{\phi(N)}} - \ket{\Psi_{\Delta_1,0}^{\phi(N)}}\|^2 .
\end{align}
Therefore we can bound the left hand side of this equality by studying its right hand side. By the triangle inequality,
\begin{align}
\|\ket{\Psi_{\Delta_0,0}^{\phi(N)}}-\ket{\Psi_{\Delta_1,0}^{\phi(N)}}\|\leq \|\ket{\Psi_{\Delta_0,0}^{\phi(N)}}-\ket{\Psi_{\Delta_0}^{\phi(N)}}\|+\|\ket{\Psi_{\Delta_0}^{\phi(N)}}-\ket{\Psi_{\Delta_1}^{\phi(N)}}\|+\|\ket{\Psi_{\Delta_1}^{\phi(N)}}-\ket{\Psi_{\Delta_1,0}^{\phi(N)}}\|.
\end{align}
By a similar argument to the one from the previous proof, the first and third term go to zero in the large $N$ limit. For the second term, by Markov's inequality,
\begin{align}
\|\ket{\Psi_{\Delta_0}^{\phi(N)}}-\ket{\Psi_{\Delta_1}^{\phi(N)}}\|^2\leq \mathbb{P}(E>\Delta_1)\leq \frac{E_0}{\Delta_1}.
\end{align}
We conclude that the $\ket{\Psi_{\Delta}}$ form a Cauchy family from taking $N\rightarrow\infty$.
\end{proof}

Now that the existence of the limit $\ket{\Psi}$ of the $\ket{\Psi_{\Delta}}$ as $\Delta\rightarrow\infty$, as a vector in the large $N$ vacuum Hilbert space, has been established, it remains to be shown that correlation functions in the states $\ket{\Psi^{\phi(N)}}$ actually converge towards those in $\ket{\Psi}$.

\begin{lem}
For $(S_N)$ a uniformly bounded sequence of finite $N$ observables asymptoting to a bounded vacuum sector observable $S$,
\begin{align}\bra{\Psi^{\phi(N)}}S_{\phi(N)}\ket{\Psi^{\phi(N)}}\rightarrow \bra{\Psi}S\ket{\Psi}.\label{eq:subconv}\end{align}
\end{lem}

\begin{proof}
Since by Lemma \ref{lem:cauchy} the family $(\ket{\Psi_\Delta})$ is Cauchy for the Hilbert space norm and the large $N$ vacuum sector of states is complete (as a Hilbert space), this family admits a limit $\ket{\Psi}$. Now let $\varepsilon>0$, and let $(S_N)$ be a uniformly bounded sequence of operators that asymptotes to a bounded operator $S$ in the vacuum sector. We are interested in bounding the quantity $\lvert\psi^{\phi(N)}(S_{\phi(N)})-\psi(S)\rvert$ (once again using the notation $\psi(X)$ for $\bra{\Psi}X\ket{\Psi}$). By the triangle inequality, for all choices of $\Delta$,
\begin{align}
\nonumber\lvert\psi^{\phi(N)}(S_{\phi(N)})-\psi(S)\rvert&\leq\lvert\psi^{\phi(N)}(S_{\phi(N)})-\psi_\Delta^{\phi(N)}(S_{\phi(N)})\rvert\\&+\lvert\psi_\Delta^{\phi(N)}(S_{\phi(N)})-\psi_\Delta(S)\rvert+\lvert\psi_\Delta(S)-\psi(S)\rvert.
\end{align}
By Lemma \ref{lem:truncate}, for $\Delta$ larger than some finite value $\Delta_0$ (fixed with $N$), the first term is upper bounded by $\varepsilon/3$. Moreover by Lemma \ref{lem:extract}, for $N$ larger than some $N_0$, the second term is upper bounded by $\varepsilon/3$. Finally by Lemma \ref{lem:cauchy}, for $\Delta>\Delta_1$ fixed, the third term is upper bounded by $\varepsilon/3$. Choosing $\Delta>\mathrm{max}(\Delta_0,\Delta_1)$, we obtain that for $N>N_0$,
\begin{align}
\lvert\psi^{\phi(N)}(S_{\phi(N)})-\psi(S)\rvert&\leq\varepsilon.
\end{align}
This shows 
\begin{align}\psi^{\phi(N)}(S_{\phi(N)})\rightarrow \psi(S).\end{align}
\end{proof}

\subsection{Step 5: exploit the uniqueness of the lmit}

So far only the convergence of the subsequence $\bra{\Psi^{\phi(N)}}S_{\phi(N)}\ket{\Psi^{\phi(N)}}$ has been established. But by Assumption 1, the sequence $\bra{\Psi^{(N)}}S_{N}\ket{\Psi^{(N)}}$ (without the extractor) admits a limit as $N\rightarrow\infty$. By uniqueness of the limit we conclude:
\begin{align}\bra{\Psi^{(N)}}S_N\ket{\Psi^{(N)}}\rightarrow \bra{\Psi}S\ket{\Psi}.\end{align}
It remains to be shown that the state $\ket{\Psi}$ is pure. But this is obvious as it is a vector in the large $N$ vacuum Hilbert space. This concludes the proof of the theorem.

\subsection{Comments on the proof}

It is worth reflecting on why the proof presented above works. In particular, the fact that several methods from mathematical analysis can be applied to the case at hand, whereas they cannot in various other examples of interest, sheds light on the physical mechanism underlying the result. This section offers some general comments on the physical interpretation of various techniques applied in the proof, and why they would not work in many other cases of interest in holography. Possible generalizations of the proof are also envisioned.

\begin{enumerate}

\item{\textbf{Dominated convergence and the entanglement wedge.}} The basic idea of the proof of Theorem \ref{thm:main} is strongly related to one of the central results in Lebesgue integration theory: the dominated convergence theorem. In functional analysis, it is often interesting to study sequences of functions $(f_N)$, which may be defined either on a discrete space or on a continuous space $X$. It is often the case that the sequence $(f_N)$ has a limit in the sense of pointwise convergence, i.e.
\begin{align}
\forall x\in X,\, f_N(x)\rightarrow f(x).
\end{align}
In general, this fact is \textbf{not} enough to conclude that
\begin{align}
\int_X f_N(x)dx\rightarrow \int_X f(x)dx.
\label{eq:intconv}
\end{align}
One can think of many counterexamples: for example, the family $(f_N)$ defined on $\mathbb{R}$ by \begin{align}
f_N(x)= 1 \;\text{if}\; x\in[N,N+1], 0 \;\text{else,}
\end{align}
satisfies, for all $x\in \mathbb{R}$,
\begin{align}
f_N(x)\rightarrow 0,
\end{align}
and yet for all $N$,
\begin{align}
\int_{-\infty}^\infty f_N(x)dx = 1.
\end{align}
However, the dominated convergence theorem tells us we actually \textbf{can} conclude that the integral converges under an additional summability condition: if there exists a function $g$ on $X$ such that  
\begin{align}
\forall x\in X,\,\;\forall N\in\mathbb{N},\; \lvert f_N(x)\rvert \leq g(x),
\end{align}
and 
\begin{align}
\int_X g(x)dx<\infty,
\end{align}
then we can actually conclude that the integral convergence \eqref{eq:intconv} holds. This condition is of course obviously violated in the above counterexample, since the $f_N$'s take a constant positive value on an interval of width 1, arbitrarily far from the origin.

In Step 2 of the proof of the theorem, a subset of indices $\{\phi(N)\}$ for which all of the $c_{AB}^{\phi(N)}$ converge towards a limit $c_{AB}^{0}$ was extracted. In general, this would not be enough to conclude that the sequence of correlation functions in the $\ket{\Psi^{\phi(N)}}$ converge towards the ones in a state where the $c_{AB}^{\phi(N)}$ have been replaced by the $c_{AB}^{0}$. However, since it was also shown that the state can be effectively truncated to a microcanonical window of finite size independent of $N$, and the $c_{AB}^{(N)}$ have a modulus upper bounded by $1$, it is possible to apply the dominated convergence theorem to the $c_{AB}^{(N)}$ and conclude that the correlation functions converge towards the ones in a state where the $c_{AB}^{0}$ are taken to be constant.

This situation should be contrasted with the one that would have arisen had we considered a version of \eqref{eq:ASSState} where the $c_{AB}^{(N)}$ are taken \textbf{above} the Hawking--Page temperature. For such states, for any fixed value of $AB$, $c_{AB}^{(N)}$ is expected to converge to zero. This is because the support of $\ket{\Psi^{(N)}}$ gets pushed to infinity at large $N$, just like in the simple counterexample presented above. Most of the support of the state $\ket{\Psi^{(N)}}$ is not restricted to a microcanonical window of fixed size close to the ground state, but rather to a window of energy of $O(N^2)$. In particular, this window is $N$-dependent, and it would not be right to approximate the correlations in the state $\ket{\Psi^{(N)}}$ with the ones in a state where the $c_{AB}^{(N)}$ have been given a value that is independent of $N$. It is this failure of applicability of the dominated convergence theorem that makes the situation much richer for states of higher energy, in particular it is no longer forbidden for the state $\ket{\Psi^{(N)}}$ to asymptote to a mixed state in the causal wedge at large $N$, and in turn for a bulk dual that is larger than the causal wedge - for example, containg the throat of a long wormhole. These kinds of violations of the assumptions of the dominated convergence theorem are intimately related to the exotic nature of the large $N$ limit at high temperature mentioned in \cite{Wit21b,SchWit22}.\footnote{I am grateful to Hong Liu for pointing out the relevance of the second reference.}

The lesson to draw from this informal discussion is that the dominated convergence theorem makes the properties of the large $N$ limit extremely constrained for states of $O(N^0)$ energy. For states of parametrically large energy, the situation is much richer.

\item{\textbf{Potential generalizations}.\footnote{I am grateful to Netta Engelhardt for discussions on this point.} The proof of Theorem \ref{thm:main} shows that a family of states that remains an $O(N^0)$ amount of energy away from the vacuum and that has a large $N$ limit on single trace operators must actually be representable as a pure state in the vacuum sector. It would be interesting to try to run a similar logic for other sectors. For example, one could try to show that a suitably defined $O(1)$ perturbation of a given large $N$ black hole background cannot lead to the emergence of a baby universe behind the horizon. An argument backing up this conjecture was actually given in \cite{EngGes25a} in the context of the simple entropy proposal for a black hole formed from fast collapse.}
\end{enumerate}

\section{How to evade the conclusion: observers, averaging, and all that}
\label{sec:Discuss}

The theorem proven above confirms, at a mathematical level of rigor, that a closed universe geometry cannot emerge in the large $N$ limit from a sequence of finite $N$ states of $O(1)$ energy. This means that in order to fit closed cosmologies into the framework of the AdS/CFT correspondence, we need to change some rules of the game. In particular, a possible interpretation of this result is that the closed universe fails to emerge \textbf{to a boundary observer who applies the traditional rules of AdS/CFT}. Such a boundary observer would be able to detect that the causal wedge is in an overall pure state. On the other hand, to an observer that is part of the closed universe, semiclassical physics may still emerge, which is reminiscent of the old idea of complementarity \cite{SusTho93}, as recently noted in \cite{EngGes25b}. 

Then, a question arises: what is the CFT interpretation of introducing such an internal ``observer's perspective"? Whatever it is, it must do something to evade the conclusion of the theorem presented in this work. In other words, restoring the semiclassicality of the baby universe must be done by lifting some of the assumptions of the theorem. Let us describe two proposed ways of restoring the semiclassicality of the baby universe from the CFT perspective, and explain which assumptions of the theorem these options lift.

\begin{enumerate}
\item \textbf{Coarse-graining over heavy operators.} In \cite{AntSas25}, it was pointed out that one way to possibly restore the semiclassicality of the baby universe is to coarse grain over the microscopic details of the operator $\mathbb{O}$ in \eqref{eq:ASSState}. In the context of the large $N$ limit emphasized in the present paper, this means that the correlation functions at finite $N$ are not computed in the fixed state defined in \eqref{eq:finiteN}, but rather, averaged over many finite $N$ states like \eqref{eq:finiteN}. The effect of this averaging is that the correlators are actually computed in a mixed state given by a classical superposition of many pure states of the form \eqref{eq:finiteN}. Since Theorem \ref{thm:main} assumes that each finite $N$ state is pure, this method evades its conclusion by computing the correlators in a mixed state defined through coarse-graining over $\mathbb{O}$ at each finite $N$. The price to pay is a modification of which states are identified between the bulk and the boundary: the large $N$ state is now identified with a family of mixed states rather than the family of pure states of the form \eqref{eq:ASSState}.
\item \textbf{Averaged large $N$ limits.} If one does not wish to change the family of finite $N$ states that the large $N$ geometry is identified with, another possible course of action is to generalize what is meant by large $N$ limit. Theorem \ref{thm:main} assumes a conventional notion of large $N$ limit: all the finite $N$ correlators of finite $N$ single trace operators converge to a fixed value as $N\rightarrow\infty$. However, it is possible to ask for a weaker requirement: for example, one can decide to only require that the correlators converge \textbf{when averaged over a sufficiently large window of values of $N$}. This is the approach undertaken in the work of Liu \cite{Liu25}, set to appear on the same day as the present paper. By introducing a family of finite $N$ states which do not have a limit as $N\rightarrow\infty$ in terms of correlation functions, but do when averaged over a sufficiently large window of values of $N$, it is possible to show that such averages lead to large $N$ correlators consistent with a mixed state in the causal wedge. It will be argued in \cite{Liu25} that generic sequences of finite $N$ low temperature partially entangled thermal states actually only converge on average, and that the average limit is consistent with a semiclassical description containing a closed universe. By only requiring a convergence on average in Theorem \ref{thm:main}, it is possible to evade its conclusion without changing the family of finite $N$ states under consideration. The challenge for such an approach is then to explain how one can systematically introduce $N$-averaging in the AdS/CFT correspondence while preserving gravitational perturbation theory in the bulk. 
\end{enumerate}

The two methods mentioned above, based on averaging either over the details of the operator $\mathbb{O}$ at each finite $N$ or over the value of $N$ itself, represent a departure from the traditional way we have been thinking about holography. In most ususal treatments of AdS/CFT, it is assumed that a \textbf{fixed} state at a \textbf{fixed} value of $N$ can be well-approximated by the infinite $N$ theory, up to errors that are suppressed in $1/N$. In particular, in holographic code models, the boundary state, as well as the value of $N$, are fixed \cite{FauLi22,Ges23}. Yet, Theorem \ref{thm:main} shows that such departures from traditional AdS/CFT are unavoidable in order to describe the experience of a semiclassical observer in a closed universe while not completely giving up on a description based on a family of finite $N$ CFTs. The relationship between the two above methods and the ``observer rules" introduced in \cite{HarUsa25,AbdSte25,AkeBue25}, which have also been shown to be able to restore a semiclassical description of the baby universe \cite{EngGes25b,AbdSte25}, is likely to be intimate, but remains to be explored in greater detail. More generally, the general principles justifying the methods described above have not been discovered yet. It is an exciting time for holography!

\section*{Acknowledgements}
It is a pleasure to thank Netta Engelhardt and Hong Liu for extensive discussions and encouragement to publish this work. I would also like to thank Suzanne Bintanja, Kasia Budzik, Netta Engelhardt, Daniel Harlow, Hong Liu, Leo Shaposhnik, Mykhaylo Usatyuk and Wayne Weng for very helpful discussions and comments on a draft of this paper. This work was made possible by the support of the Gordon and Betty Moore Foundation via the Black Hole Initiative, the Heising Simons Foundation under grant number 2023-4430, and the Department of Energy Office of Science under Early Career Award DE-SC0021886.\\

\bibliographystyle{jhep}
\bibliography{all}
\end{document}